\documentclass[11pt,english]{article}
\usepackage[T1]{fontenc}
\usepackage[latin9]{inputenc}
\usepackage[letterpaper]{geometry}
\usepackage{amsthm}
\usepackage{amsmath}
\usepackage{amssymb}

\makeatletter
\theoremstyle{plain}

\makeatother

\newtheorem{thm}{Theorem}  
\newtheorem{cor}{Corollary} \newtheorem{lem}{Lemma}
\newtheorem{prop}{Proposition}  

\def\ket{\rangle}

\def\be{\begin{eqnarray}}
\def\ee{\end{eqnarray}}
\def\bee{\begin{eqnarray*}}
\def\eee{\end{eqnarray*}}
\def\ot{\otimes}
\def\tr{\hbox{Tr}}

\def\lnrm{\left \vert \left \vert}
\def\rnrm{\right \vert \right \vert}

\makeatother

\makeatother

\makeatother

\makeatother

\makeatother

\makeatother

\usepackage{babel}

\makeatother

\usepackage{babel}

\makeatother

\usepackage{babel}

\makeatother

\usepackage{babel}

\makeatother

\usepackage{babel}

\makeatother

\usepackage{babel}

\begin{document}

\title{Tight bounds on the distinguishability of quantum states under separable
measurements}

\author{Somshubhro Bandyopadhyay%
\thanks{Department of Physics and Center for Astroparticle Physics and Space
Science, Bose Institute, Block EN, Sector V, Bidhan Nagar, Kolkata
700091, India; Email: som@jcbose.ac.in%
} $\;\;$and Michael Nathanson%
\thanks{Department of Mathematics and Computer Science, Saint Mary's College
of California, Moraga, CA, 94556, USA; Email: man6@stmarys-ca.edu %
}}
\maketitle
\begin{abstract}
One of the many interesting features of quantum nonlocality is that
the states of a multipartite quantum system cannot always be distinguished
as well by local measurements as they can when all quantum measurements
are allowed. In this work, we characterize the distinguishability
of sets of multipartite quantum states when restricted to separable
measurements, those which contain the class of local measurements
but nevertheless are free of entanglement between the component systems.
We consider two quantities: the separable fidelity, a truly quantum
quantity, which measures how well we can {}``clone'' the input state,
and the classical probability of success, which simply gives the optimal
probability of identifying the state correctly.

We obtain lower and upper bounds on the separable fidelity and give
several examples in the bipartite and multipartite settings where
these bounds are optimal. Moreover the optimal values in these cases
can be attained by local measurements. We further show that for distinguishing
orthogonal states under separable measurements, a strategy that maximizes
the probability of success is also optimal for separable fidelity.
We point out that the equality of fidelity and success probability
does not depend on an using optimal strategy, only on the orthogonality
of the states. To illustrate this, we present an example where two
sets (one consisting of orthogonal states, and the other non-orthogonal
states) are shown to have the same separable fidelity even though
the success probabilities are different. 
\end{abstract}

\section{Introduction}

Suppose a composite quantum system is known to be in one of many states,
not necessarily orthogonal, such that its parts are distributed among
spatially separated observers. The goal is to learn about the state
of the system using only local quantum operations and classical communication
between the parties (LOCC). This problem, known as local state discrimination,
is of considerable interest \cite{Walgate-2000,Walgate-2002,Peres-Wootters-1991,Bennett-I-99,fan-2005,Horodecki-2003,Nathanson-2005,Watrous-2005,Duan2007,Qi-Duan,Bandyo-2011},
as in many instances the information obtainable by LOCC is strictly
less than that achieved with global measurements \cite{Bennett-I-99,Ghosh-2001,Ghosh-2002,Ghosh-2004}.
This gives rise to a new kind of nonlocality \cite{Bennett-I-99,Horodecki-2003,Bandyo-2011},
conceptually different from that captured through the violation of
Bell inequalities \cite{Bell-1964,CHSH}. Thus the problem of local
state discrimination and the phenomenon of nonlocality serve to explore
fundamental questions related to local access of global information
\cite{Peres-Wootters-1991,Ghosh-accessible information,Horodecki-accessible-information},
and the relationship between entanglement and local distinguishability
\cite{Bennett-I-99,Bandyo-2010,Horodecki-2003,Hayashi-etal-2006}.
Moreover, it has found novel applications such as data hiding \cite{Terhal2001,DiVincenzo2002,Eggeling2002}
and secret sharing \cite{Markham-Sanders-2008}.

There are many celebrated results identifying sets of states for which
perfect local discrimination is possible and sets for which it is
not. In particular, any two pure states can be optimally distinguished
with LOCC \cite{Walgate-2000,Virmani-2001} but no more than $d$
maximally entangled states on $\mathbb{C}^{d}\otimes\mathbb{C}^{d}$
can be \cite{Nathanson-2005,Cosentino-2012}. A complete basis of
a composite space which can be distinguished with separable measurements
must be a product basis but this condition is not sufficient in general
\cite{Bennett-I-99,Ghosh-2001,Horodecki-2003}. Finally, sometimes
increasing the average entanglement in a set can enable state discrimination
\cite{Horodecki-2003}. More recent studies include distinguishing
states (pure or mixed) when many copies are provided \cite{Calsamiglia,Bandyo-2011,Yu-Duan,Nathanson-2010,MatthewsWehnerWinter}.

The class of LOCC measurements does not have a simple mathematical
characterization, and optimization is often analytically intractable.
In this paper, we will focus on the class of separable measurements,
those which are free of entanglement between the component systems.
These comprise a strict superset of LOCC measurements and are much
more amenable to analytic results (as in \cite{Duan-2009,Navascu\'es}).
It should be noted, however, that while every LOCC protocol can be
realized by a rank-one separable measurement, the converse is known
not to be true \cite{Bennett-I-99,Bennett-II-99 +Divin-2003}.

The focus of this paper is in quantifying imperfect local discrimination,
a question which has been settled in the case of a pair nonorthogonal
pure states \cite{Virmani-2001} but has generally not been explored
as deeply. In \cite{Nathanson-2005} bounds on the error probability
in distinguishing bipartite orthogonal states were obtained, and in
\cite{Cosentino-2012} upper bounds on the maximum probability of
perfect local discrimination were derived for special sets of maximally
entangled states. In a different approach, a complementary relation
between locally accessible information and final average entanglement
was observed \cite{Ghosh-accessible information,Horodecki-accessible-information}
which provides upper bounds on the locally accessible information
and is known to be optimal for some classes of states. Other approaches
used measurements with positive partial transpose \cite{Cosentino-2012,Yu-Duan};
the set of such measurements contains the separable ones as a strict
subset.

We will use two measures of distinguishability, the average fidelity
and the success probability. The notion of \emph{average fidelity}
was first considered by Fuchs and Sasaki to measure the {}``quantumness''
of a set of states \cite{Fuchs=Sasaki,Fuchs}. The authors imagine
a quantum source which emits a quantum state from $S=\left\{ p_{i},|\psi_{i}\rangle\right\} $,
headed towards a receiver. The state is intercepted by Eve, who performs
a complete measurement on the state. She sends the results of her
measurement to her partner Yves via a classical channel; he uses this
classical information to construct a state $\vert\phi_{a}\rangle$
(here, the subscript $a$ indicates the measurement outcome of Eve)
which is sent on to the original intended receiver. The average fidelity
measures the probability that the activity of Eve and Yves will not
be detected. (By contrast, the probability of correct state identification
measures Yves's ability to correctly identify which state was sent.)
If the set of states is {}``highly quantum,'' then passing through
the classical channel is necessarily very disruptive, the eavesdropping
is detectable, and the average fidelity is low. If the states are
less quantum, then the classical restriction is less disruptive and
the average fidelity is higher.

In our restricted problem of local state discrimination, the objective
is to maximize the average fidelity over all of Eve\textquoteright{}s
measurements which are separable, yielding the \emph{separable fidelity}
\cite{Navascu\'es}. The idea is that once we broadcast the classical
information gleaned from the measurement, anyone can use it to prepare
a \textquotedbl{}best guess\textquotedbl{} state so as to maximize
the average fidelity. In particular, we can now assume that all the
components are in the same location, so all global operations are
allowed. We derive lower and upper bounds on the separable fidelity
and provide examples in bipartite and multipartite settings where
the bounds are shown to be optimal. This is shown by an explicit local
strategy for each example.

The second figure of merit that we consider is the probability of
identifying the state which was prepared. Note that, while the fidelity
is a measure of quantumness, the probability of success is a classical
measure of how well a quantum protocol encodes and decodes classical
information. We show that, when the states are mutually orthogonal,
the separable fidelity coincides with the maximum success probability,
which relates our results to bounds obtained in \cite{Nathanson-2005}.
We point out that this equality between separable fidelity and probability
of success depends crucially on the orthogonality of the states. To
illustrate this, we present an example where two sets (one consisting
of orthogonal states, and the other non-orthogonal states) are shown
to have the same separable fidelity even though the success probabilities
are different.

The rest of the paper is organized as follows. Section \ref{Sec-Basics}
introduces the basic notions and defines the measures of distinguishability
of an ensemble, demonstrating the equality of the success probability
and separable fidelity when the states are orthogonal. Section \ref{Sec-Bounds}
provides upper and lower bounds on the success probability and the
separable fidelity, and Sections \ref{Sec-Examples} and \ref{Sec-QubitExamples}
contain sets of examples for which these bounds are tight. We conclude
with a discussion of open problems in Section \ref{Sec-Conclusions}.

\section{Measures of Distinguishability}

\label{Sec-Basics}

\subsection{Perfect distinguishability and rank one separable measurements}

We consider a $k$-partite quantum system with $k\geq2$. The associated
Hilbert space $\mathcal{H}$ takes the form $\mathcal{H}=\otimes_{i=1}^{k}\mathcal{H}_{i}$,
where the dimension $d_{i}$ of each local Hilbert space is finite.

A separable measurement $\mathbb{M}=\left\{ M_{1},M_{2},\cdots,M_{n}\right\} $
on $\mathcal{H}$ is a POVM satisfying $\sum_{a=1}^{n}M_{a}=I_{\mathcal{H}}$,
and for each $a$, $M_{a}$ is a separable, positive semi-definite
operator \cite{Watrous-2005,Duan-2009}. By definition each $M_{a}$
is a positive linear combination of rank-one projections onto product
states, so without loss of generality we will assume that each $M_{a}$
is of the form \begin{eqnarray}
M_{a} & = & m_{a}|\chi_{a}\rangle\langle\chi_{a}|=m_{a}|\chi_{a}^{1}\rangle\langle\chi_{a}^{1}|\otimes|\chi_{a}^{2}\rangle\langle\chi_{a}^{2}|\otimes\cdots\otimes|\chi_{a}^{k}\rangle\langle\chi_{a}^{k}|\label{M_a-rank-one}\end{eqnarray}
 where $m_{a}\in(0,1]$ and $|\chi_{a}\rangle$ is a normalized product
state in $\otimes_{i=1}^{k}\mathcal{H}_{i}$.

We then associate each measurement outcome with the most likely input
state to produce it. This defines a decoding function $\mathbb{G}$
with $\mathbb{G}(a)\in\mathcal{H}$. Combining this with our rank-one
representation of separable POVMs gives us the necessary and sufficient
condition for perfect distinguishability by a separable measurement.

\begin{prop} The pure quantum states $|\psi_{1}\rangle,|\psi_{2}\rangle,\ldots,|\psi_{k}\rangle$
are perfectly distinguishable by a separable measurement if and only
if there exists a separable POVM $\mathbb{M}=\left\{ M_{1},M_{2},\cdots,M_{n}\right\} $
and a decoding function $\mathbb{G}$ such that each $M_{a}$ is rank
one in the form (\ref{M_a-rank-one}) and for all $i$ and $j$\emph{\begin{eqnarray}
\langle\psi_{j}\left|M\left(i\right)\right|\psi_{j}\rangle & = & \delta_{ij}\label{prop1-eq-1}\end{eqnarray}
}  where ${\displaystyle M(i)=\sum_{a:\mathbb{G}(a)=i}M_{a}}$.\end{prop}

\subsection{Separable fidelity }

Given a set $S=\left\{ p_{i},|\psi_{i}\rangle\right\} $ of pure multipartite
quantum states $|\psi_{i}\rangle$ occurring with probabilities $p_{i}$,
we are often unable to distinguish them perfectly. In particular,
even if they are orthogonal it is not sufficient to guarantee distinguishability
with separable measurements. In this case, we wish to quantify how
much can be learned about the state of our system. The average fidelity
is one such measure, calculated with respect to a particular physical
protocol and information processing scheme, defined initially in \cite{Fuchs=Sasaki}.
Thus, for fixed set $S=\left\{ p_{i},|\psi_{i}\rangle\right\} $,
a measurement (POVM) $\mathbb{M}=\left\{ M_{a}\right\} $, and a guessing
strategy $\mathbb{G}:a\rightarrow|\phi_{a}\rangle$, the average fidelity
is given by \cite{Fuchs=Sasaki,Navascu\'es},\begin{eqnarray}
F\left(\mathbb{M},\mathbb{G}\right) & = & \sum_{i,a}p_{i}\left\langle \psi_{i}\left|M_{a}\right|\psi_{i}\right\rangle \left|\left\langle \psi_{i}|\phi_{a}\right\rangle \right|^{2}.\label{avg-fidelity}\end{eqnarray}
 This measures our ability to prepare a new quantum system in a state
which is close to the original state $\vert\psi_{i}\ket$. Note that
$0\leq F\left(\mathbb{M},\mathbb{G}\right)\leq1$, and $F\left(\mathbb{M},\mathbb{G}\right)=1$
if and only if the procedure $\left(\mathbb{M},\mathbb{G}\right)$
identifies the given state of our system perfectly.

In our work, we wish to understand the limitations of using only separable
measurements to distinguish the elements of $S=\left\{ p_{i},|\psi_{i}\rangle\right\} $.
Thus, the primary quantity of interest is the optimized form of $F\left(\mathbb{M},\mathbb{G}\right)$,
where the optimization is over all separable measurements and guessing
strategies. In Ref. \cite{Fuchs=Sasaki} (see also \cite{Fuchs})
the authors introduced the concept of {\em achievable fidelity}
for a fixed measurement $\mathbb{M}$, obtained by optimizing over
all guessing strategies $\mathbb{G}$: \begin{eqnarray}
F\left(\mathbb{M}\right) & = & \sup_{\mathbb{G}}F\left(\mathbb{M},\mathbb{G}\right).\label{achievable-sep-F}\end{eqnarray}
 Thus the achievable fidelity gives the best possible fidelity for
a given measurement $\mathbb{M}$.

The \emph{separable fidelity} is therefore defined as \cite{Navascu\'es}:\begin{eqnarray}
F_{S} & = & \sup_{\mathbb{M}}F\left(\mathbb{M}\right)\nonumber \\
 & = & \sup_{\mathbb{M},\mathbb{G}}F\left(\mathbb{M},\mathbb{G}\right)\label{sep-fidelity}\end{eqnarray}
 where the supremum is taken over all {\em separable} measurements
$\mathbb{M}$ and decoding schemes $\mathbb{G}$. In Ref.$\,$\cite{Navascu\'es}
it was shown that the separable fidelity can be obtained as the limit
of a sequence of real numbers, that is, $F_{S}=\lim_{n\rightarrow\infty}F_{S}^{\left(n\right)}$,
where $F_{S}^{(1)}\geq F_{S}^{(2)}\geq...$ and each $F_{S}^{(i)}$
can be efficiently computed numerically. While the result in \cite{Navascu\'es}
guarantees asymptotic convergence, it is however unclear how many
iterations it might take.

Note that for a measurement $\mathbb{M}$ which is separable, achievable
fidelity is by definition an intermediate quantity between the average
fidelity and separable fidelity. The advantage of introducing the
notion of achievable fidelity is that it can be computed \emph{exactly}
for any measurement. As we will show, the achievable fidelity is the
key ingredient in our analysis towards obtaining the desired bounds
on the separable fidelity.

\subsection{Separable fidelity with orthogonal states}

In the special case in which the elements of $S$ are mutually orthogonal,
there is an especially straightforward way to calculate the achievable
fidelity:

\begin{lem} For a set $S=\left\{ p_{i},|\psi_{i}\rangle\right\} _{i=1}^{i=N}$
of orthogonal pure quantum states, and a measurement $\mathbb{M}=\left\{ M_{a}\right\} $,
the achievable fidelity is given by \begin{eqnarray}
F\left(\mathbb{M}\right) & = & \sum_{a}\mu_{a},\label{lemma-1-eq-1}\end{eqnarray}
 where \emph{\begin{eqnarray}
\mu_{a} & = & \max_{i}\left\{ p_{i}\left\langle \psi_{i}\left|M_{a}\right|\psi_{i}\right\rangle \right\} .\label{lemma-1-eq-2}\end{eqnarray}
}  \end{lem} \begin{proof} Observe that for any given measurement
$\mathbb{M}$, and an associated guessing strategy $\mathbb{G}$,
the average fidelity can be written as \begin{eqnarray}
F\left(\mathbb{M},\mathbb{G}\right) & = & \sum_{a}\left\langle \phi_{a}\right|\left(\sum_{i=1}^{N}p_{i}\left\langle \psi_{i}\left|M_{a}\right|\psi_{i}\right\rangle \left|\psi_{i}\right\rangle \left\langle \psi_{i}\right|\right)\left|\phi_{a}\right\rangle .\label{lemma-1-eq-3}\\
 & \le & \sum_{a}\lnrm\sum_{i=1}^{N}p_{i}\left\langle \psi_{i}\left|M_{a}\right|\psi_{i}\right\rangle \left|\psi_{i}\right\rangle \left\langle \psi_{i}\right|\rnrm_{\infty}\label{lemma-1-eq-3B}\end{eqnarray}
 where the norm $\lnrm\cdot\rnrm_{\infty}$ is simply the largest
singular value. By letting $|\phi_{a}\ket$ equal the dominant eigenvector,
we can achieve the optimum in Eq.$\,$(\ref{lemma-1-eq-3B}), in which
case we will have $F(\mathbb{M})$. Noting that the operator $\sum_{i,}p_{i}\left\langle \psi_{i}\left|M_{a}\right|\psi_{i}\right\rangle \left|\psi_{i}\right\rangle \left\langle \psi_{i}\right|$
is diagonal in the orthogonal states $|\psi_{i}\rangle$, the dominant
eigenvector is one of the $\{\vert\psi_{i}\ket\}$ and the result
follows immediately. Thus for every $a$, the best guess state $|\phi_{a}\rangle$
is simply the state $|\psi_{k}\rangle$ such that $p_{k}\left\langle \psi_{k}\left|M_{a}\right|\psi_{k}\right\rangle =\max_{i}\left\{ p_{i}\left\langle \psi_{i}\left|M_{a}\right|\psi_{i}\right\rangle \right\} $.
\end{proof} 

When the states in $S$ are mutually orthogonal, we sometimes treat
the quantum state $\vert\psi_{i}\ket$ as simply an encoding of the
classical label {}``i,'' and our goal in state discrimination is
simply to recover this value. The most natural measure for this is
the probability of success. If our system is initially in the state
$\vert\psi^{*}\ket$ taken from $S=\{p_{i},\vert\psi_{i}\ket\}$,
then the probability of successful identification using the measurement
$\mathbb{M}$ is defined as \begin{eqnarray}
P_{s}(\mathbb{M}) & = & P\left(\vert\hat{\psi}\ket=\vert\psi^{*}\ket\right)\\
 & = & \sum_{i}p_{i}P\left(\left.|\hat{\psi}\rangle=\vert\psi_{i}\rangle\right\vert \vert\psi^{*}\ket=\vert\psi_{i}\ket\right)\end{eqnarray}
 where $\vert\hat{\psi}\ket\in S$ is our best guess after performing
the measurement $\mathbb{M}$. We will write $P_{s}$ to indicate
the optimal value of $P_{s}(\mathbb{M})$ over all measurements $\mathbb{M}$
and $P_{s}(S)$ for the optimal value of $P_{s}(\mathbb{M})$ over
all separable $\mathbb{M}$.

The authors of \cite{Fuchs=Sasaki} showed that for any measurement
$\mathbb{M}$, $P_{s}(\mathbb{M})\le F(\mathbb{M})$. When the elements
of $S$ are mutually orthogonal, we know that there exists a (not
necessarily separable) measurement protocol $\mathbb{M}^{*}$ so that
$P_{s}(\mathbb{M}^{*})=F(\mathbb{M}^{*})=1$. This good fortune sometimes
obscures the fact that the equality of $P_{s}$ and $F$ does not
depend on using an optimal measurement, only on the orthogonality
of the states. This fact (alluded to in \cite{Navascu\'es}) allows
us to see that the strategy that minimizes the probability of error
is also optimal for separable fidelity.

\begin{thm} Let $S=\left\{ p_{i},|\psi_{i}\rangle\right\} $ be a
set of mutually orthogonal pure multipartite quantum states. Then
for any measurement $\mathbb{M}$, \emph{\begin{eqnarray}
P_{s}(\mathbb{M}) & = & F\left(\mathbb{M}\right)\end{eqnarray}
}  In particular, if we optimize over all all separable measurements
$\mathbb{M}$, we have \emph{\begin{eqnarray}
P_{s}(S) & = & F_{S}\end{eqnarray}
}  and these maxima are achieved using the same optimal measurement.
\end{thm} \begin{proof} The proof follows by noting that, when calculating
$P_{s}(\mathbb{M})$, our guess should be the state $\vert\psi_{i}\rangle$
which has the maximum likelihood conditioned on the observed outcome
$a$, which we write as $p\left(i|a\right)$. Thus we can rewrite
$P_{s}(\mathbb{M})$ by conditioning on the measurement outcome to
get \begin{eqnarray}
P_{s}(\mathbb{M}) & = & \sum_{a}p(a)\max_{i}p\left(i|a\right)\\
 & = & \sum_{a}\max_{i}p_{i}\langle\psi_{i}|M_{a}|\psi_{i}\rangle=\sum_{a}\mu_{a},\end{eqnarray}
 which is the expression for $F(\mathbb{M})$ from Lemma 1. 

To prove that $P_{s}(S)=F_{S}$ and that these are achieved using
the same separable measurement, we proceed in the following way. Let
$\mathfrak{M}$ be the set of separable measurements, and find optimal
measurements $\mathbb{M}_{1}$ and $\mathbb{M}_{2}$ in $\mathfrak{M}$
such that for all $\mathbb{M}\in\mathfrak{M}$, \begin{eqnarray*}
P_{s}\left(\mathbb{M}\right) & \leq & P_{s}\left(\mathbb{M}_{1}\right)\\
F_{S}\left(\mathbb{M}\right) & \leq & F_{S}\left(\mathbb{M}_{2}\right)\end{eqnarray*}
 Since the states are orthogonal, we know that $P_{s}\left(\mathbb{M}_{i}\right)=F_{S}\left(\mathbb{M}_{i}\right)$
for $i=1,2$, which in turn implies that \[
F_{S}\left(\mathbb{M}_{2}\right)=P_{s}\left(\mathbb{M}_{2}\right)\leq P_{s}\left(\mathbb{M}_{1}\right)=F_{S}\left(\mathbb{M}_{1}\right)\leq F\left(\mathbb{M}_{2}\right).\]
 Thus, the optimal measurement in $\mathfrak{M}$ for $P_{s}$ is
also optimal for $F$ and vice versa, and \begin{eqnarray}
P_{s}\left(S\right) & = & F_{S}\label{min-error-eq-6}\end{eqnarray}
 Note that we have not used any specific properties of separable measurements
except the existence of optimal measurements. This shows that the
argument works for any compact set of measurements $\mathfrak{M}$.
\end{proof}

\section{Bounds on Separable Fidelity}

\label{Sec-Bounds}

We will first obtain lower and upper bounds on the separable fidelity.
Later, we will give examples where these bounds are shown to be optimal.
It may be noted that an upper bound on the separable fidelity is also
an upper bound on the optimal local fidelity, that is, the best possible
fidelity attainable by LOCC.

\subsection{Lower bounds}

For the set $S=\left\{ p_{i},|\psi_{i}\rangle\right\} _{i=1}^{i=N}$,
consider the collection of subsets of $S$ that are perfectly distinguishable
by separable operations. That is, \begin{eqnarray}
\mathcal{R} & = & \left\{ X\subset S:F_{S}\left(X\right)=1\right\} \label{eq:defn-R}\end{eqnarray}
 If $S$ contains a pair of orthogonal states, then this two-element
set is in $\mathcal{R}$ since any two orthogonal pure states can
always be perfectly distinguished by LOCC \cite{Walgate-2000}. Let
$P(X)$ be the \emph{a priori} probability that a state selected from
$S$ is an element of $X$; that is, $P(X)=\sum_{|\psi_{i}\rangle\in X}p_{i}$.
Note that two such sets $X_{1}$ and $X_{2}$ need not be disjoint.

\begin{thm} Let $S=\left\{ p_{i},|\psi_{i}\rangle\right\} $ be a
set of pure multipartite quantum states. Then, \emph{\begin{eqnarray}
F_{S} & \geq & \max_{X\in{\cal R}}\{P(X)\}.\label{Theorem-2-eq-1}\end{eqnarray}
}  \end{thm}

The proof is fairly immediate based on two observations. The first
is that the separable fidelity is lower bounded by the success probability,
as shown in \cite{Fuchs=Sasaki}. The second is that once we know
we are in the set $X$, then the probability of successfully identifying
our state is 1. For completeness, we include the following calculation.
\begin{proof} Let $X=\left\{ |\psi_{i}\rangle\right\} \in{\cal R}$
be any set whose elements are perfectly distinguishable by a separable
measurement $\mathbb{M}^{X}$. Let us denote the elements of $\mathbb{M}^{X}$
by $M_{q}^{X}$, where each element has the property that \begin{eqnarray}
\langle\psi_{r}|M_{q}^{X}|\psi_{r}\rangle=\langle\psi_{r}|\psi_{q}\rangle=\delta_{rq} &  & \;\forall~|\psi_{r}\rangle\in X\label{Thm-2-eq-2}\end{eqnarray}
 Since the elements of $X$ are distinguishable, they must be mutually
orthogonal, a property which need not be shared by the entire set
$S$. The rest of (\ref{Thm-2-eq-2}) follows from Proposition 1.

We shall now bound the average separable fidelity for the states in
$S$ by considering the strategy that consists of the measurement
$\mathbb{M}^{X}$ and the guessing map $\mathbb{G}^{X}:M_{q}^{X}\rightarrow|\psi_{q}\rangle\in X$.
With this the average separable fidelity is given by, \begin{eqnarray}
F\left(\mathbb{M}^{X},\mathbb{G}^{X}\right) & = & \sum_{i,q}p_{i}\left\langle \psi_{i}\left|M_{q}^{X}\right|\psi_{i}\right\rangle \left|\langle\psi_{q}|\psi_{i}\rangle\right|^{2}\;:|\psi_{q}\rangle\in X,|\psi_{i}\rangle\in S\nonumber \\
 & \ge & \sum_{i,q}p_{i}\left\langle \psi_{i}\left|M_{q}^{X}\right|\psi_{i}\right\rangle \left|\langle\psi_{q}|\psi_{i}\rangle\right|^{2}\;:|\psi_{q}\rangle,|\psi_{i}\rangle\in X\nonumber \\
 & = & \sum_{i,q}p_{q}\delta_{iq}=\sum_{q\in X}p_{q}=P(X)\label{Thm-2-eq-3}\end{eqnarray}
 where to arrive at the last line we have used Eq.$\,$(\ref{Thm-2-eq-2}).
The proof now follows by noting that\begin{eqnarray*}
F_{S} & \geq & \max_{X\in{\cal {R}}}F\left(\mathbb{M}^{X},\mathbb{G}^{X}\right)=\max_{X\in{\cal R}}\left\{ P(X)\right\} .\end{eqnarray*}
 \end{proof} We see that the lower bound does not depend upon the
cardinality of the set $X$, only on the \emph{a priori} probabilities.
Often we are interested in the scenario in which the states are equally
likely. Then the cardinality of the sets $X$ matter as given in the
following corollary:

\begin{cor} Let $S=\left\{ |\psi_{i}\rangle\right\} $ be a set of
$N$ mutually orthogonal pure multipartite quantum states where all
states are equally likely. Let $m$ be the maximum size of a subset
of $S$ that is perfectly distinguishable by separable measurements.
Then, \emph{\begin{eqnarray}
F_{S} & \geq & \frac{m}{N}\label{Cor-1-eq-1}\end{eqnarray}
}  \end{cor}

\subsection{Upper bounds}

We will now derive upper bounds on the separable fidelity for bipartite
systems. We want to emphasize that the bounds can be applied to multipartite
cases as well by taking the minimum of the upper bounds across all
bi-partitions. The upper bounds also have an additional significance
in that they can be used to obtain the conditions when a given set
of states cannot be perfectly distinguished. We begin by noting a
useful result proved in, e.g., \cite{Harrow,Nathanson-2010}.

\begin{lem}Let $|\psi\rangle\in\mathbb{C}^{d_{1}}\otimes\mathbb{C}^{d_{2}}$
be a bipartite pure state with Schmidt coefficients $\sqrt{\lambda_{1}}\geq\sqrt{\lambda_{2}}\geq\sqrt{\lambda_{3}}\geq\cdots\geq\sqrt{\lambda_{d_{1}}}$.
If $T$ is a measurement operator, $0\le T\le I$, which has a positive
partial transpose, then\begin{eqnarray}
\langle\psi|T|\psi\rangle & \leq & \lambda_{1}\mbox{Tr}T\label{lem-2-eq-1}\end{eqnarray}
 and this bound is tight. In particular, letting $T$ be any rank-one
separable projection $\vert\phi_{1}\rangle\langle\phi_{1}\vert\otimes\vert\phi_{2}\rangle\langle\phi_{2}\vert$,
we have \emph{\begin{eqnarray}
\max_{|\phi_{1}\rangle|\phi_{2}\rangle}\left|\langle\phi_{1}|\langle\phi_{2}|\psi\rangle\right|^{2} & = & \lambda_{1}.\label{lem-2-eq-2}\end{eqnarray}
}  \end{lem}

We note that the positive partial transpose condition is a weaker
one than separability; and that the method of \cite{Navascu\'es}
uses the fact that this class has some computational advantages. However,
in this work we will continue to focus on the class of separable measurements.
For starters, Lemma 2 gives us an immediate upper bound on the separable
fidelity.

\begin{thm} Let $S=\left\{ p_{i},|\psi_{i}\rangle\right\} $ be a
set of states in $\mathbb{C}^{d_{1}}\otimes\mathbb{C}^{d_{2}}$, where
$d_{1}\leq d_{2}$. Let $\sqrt{\lambda_{i}}$ be the largest Schmidt
coefficient of the state $|\psi_{i}\rangle$. Then, \begin{eqnarray}
F_{S} & \leq & d_{1}d_{2}||\Lambda||_{\infty},\label{theorem-3-eq-1}\end{eqnarray}
 where $\Lambda=\sum_{i}p_{i}\lambda_{i}|\psi_{i}\rangle\langle\psi_{i}|$.
\end{thm} \begin{proof} For some rank-one separable measurement
$\mathbb{M}=\left\{ M_{a}=m_{a}|\chi_{a}\rangle\langle\chi_{a}|\right\} $,
where $|\chi_{a}\rangle\in\mathbb{C}^{d_{1}}\otimes\mathbb{C}^{d_{2}}$
is a normalized product vector, and a guessing strategy $\mathbb{G}:a\rightarrow|\phi_{a}\rangle$,
the average fidelity and achievable fidelity are given by \begin{eqnarray}
F\left(\mathbb{M},\mathbb{G}\right) & = & \sum_{a}m_{a}\left\langle \phi_{a}\right|\left(\sum_{i=1}^{N}p_{i}\left|\langle\psi_{i}|\chi_{a}\rangle\right|^{2}\left|\psi_{i}\right\rangle \left\langle \psi_{i}\right|\right)\left|\phi_{a}\right\rangle \label{theorem-3-eq-2}\\
F\left(\mathbb{M}\right) & = & \sum_{a}m_{a}\lnrm\sum_{i=1}^{N}p_{i}\left|\langle\psi_{i}|\chi_{a}\rangle\right|^{2}\left|\psi_{i}\right\rangle \left\langle \psi_{i}\right|\rnrm_{\infty}\end{eqnarray}
 The second line follows since the achievable fidelity maximizes over
all choices of $\vert\phi_{a}\ket$, which will simply be the maximum
eigenvector of the indicated operator. Using Lemma 2, $\left|\langle\psi_{i}|\chi_{a}\rangle\right|^{2}\le\lambda_{i}$
and \begin{eqnarray*}
F\left(\mathbb{M}\right) & \le & \sum_{a}m_{a}\lnrm\sum_{i}p_{i}\lambda_{i}\left|\psi_{i}\right\rangle \left\langle \psi_{i}\right|\rnrm_{\infty}=||\Lambda||_{\infty}\sum_{a}m_{a}=d_{1}d_{2}||\Lambda||_{\infty}\end{eqnarray*}
 where we have used $\sum_{a}m_{a}=\tr\left(\sum m_{a}|\chi_{a}\rangle\langle\chi_{a}|\right)=d_{1}d_{2}$.
Because the above bound on the achievable fidelity holds for any measurement
including the optimal one, this completes the proof. \end{proof}
So far the bounds obtained are completely general. The following corollaries
concern two special but extensively studied cases in the literature
\cite{Ghosh-2001,Ghosh-2004,fan-2005,Nathanson-2005,Cosentino-2012,BGK,Qi-Duan}:
equally likely orthogonal states and maximally entangled states.

Note that in the special case when the states $\{|\psi_{i}\rangle\}$
are orthogonal, $||\Lambda||_{\infty}=\max_{i}\left\{ p_{i}\lambda_{i}\right\} $.
In this case, we get the following specific corollary.

\begin{cor}For a set of $N$ equally likely orthogonal states in
$\mathbb{C}^{d_{1}}\otimes\mathbb{C}^{d_{2}}$, we have \begin{eqnarray}
F_{S} & \leq & \frac{\lambda_{\max}d_{1}d_{2}}{N}\label{cor-2-eq-1}\end{eqnarray}
 \emph{ where $\lambda_{\max}=\max_{i}\lambda_{i}$.} \end{cor}

\begin{cor} Let $S=\left\{ p_{i},|\Psi_{i}\rangle:i=1,...,N\right\} $
be a set of maximally entangled states in $\mathbb{C}^{d_{1}}\otimes\mathbb{C}^{d_{2}}$,
where $d_{1}\leq d_{2}$. Then, \begin{eqnarray}
F_{S} & \leq & ||\rho||_{\infty}d_{2},\label{cor-3-eq-1}\end{eqnarray}
 where $\rho=\sum_{i}p_{i}|\psi_{i}\rangle\langle\psi_{i}|$ is the
mixed state representing our knowledge of the system prior to any
measurement. \end{cor}

Corollary 3 uses the fact that if $\vert\psi_{i}\rangle$ is maximally
entangled in $\mathbb{C}^{d_{1}}\otimes\mathbb{C}^{d_{2}}$ with $d_{1}\leq d_{2}$,
then $\lambda_{i}=\frac{1}{d_{1}}$ and $\rho=d_{1}\Lambda$. In \cite{Fuchs=Sasaki},
it was observed that $||\rho||_{\infty}$ is a weak lower bound on
the fidelity when global measurements are allowed.

\textbf{Remark 1.} It is worth noting two interesting consequences
when the above results are applied to orthogonal states. First of
all, any set of maximally entangled states in $\mathbb{C}^{d_{1}}\otimes\mathbb{C}^{d_{2}}$,
each having Schmidt rank $d_{1}$, cannot be perfectly distinguished
by LOCC (or by separable measurements) if $p_{\max}<\frac{1}{d_{2}}$.
Second, any set of $N$ equally likely maximally entangled states
in $\mathbb{C}^{d_{1}}\otimes\mathbb{C}^{d_{2}}$ cannot be perfectly
distinguished by separable measurements (and therefore by LOCC) if
$N>d_{2}$. This generalizes the known result that in $\mathbb{C}^{d}\otimes\mathbb{C}^{d}$
no more than $d$ maximally entangled states of Schmidt rank $d$
can be perfectly distinguished by LOCC \cite{Nathanson-2005} and
is a consequence of the less known fact that such a set of maximally
entangled states cannot be distinguished with a positive partial transpose
measurement \cite{Cosentino-2012,Qi-Duan}.\\

\textbf{Remark 2} Note that the bound in Eq.$\,$(\ref{cor-2-eq-1})
matches the maximum probability of distinguishing any $N$ equally
likely states in $\mathbb{C}^{d_{1}}\otimes\mathbb{C}^{d_{2}}$ by
LOCC \cite{Nathanson-2005}. Theorem 1 tells us that the success probability
and the average fidelity will be equal in this case; Corollary 2 strengthens
the result in \cite{Nathanson-2005} by extending it to all separable
measurements. However, if we apply Theorem 3 to a set of nonorthogonal
maximally entangled states, the upper bound on the fidelity \emph{increases}
while bounds on the probability of success tend to \emph{decrease}.
Note also that if the states are sufficiently nonorthogonal, the bound
in Theorem 3 can be greater than 1 and hence not informative.

Theorem 3 is useful especially in cases when all of the Schmidt coefficients
are the same, as are the \emph{a priori} probabilities. The following
theorem gives an analogous result that can be tight in more general
settings:

\begin{thm} Let $S=\left\{ p_{i},|\psi_{i}\rangle:i=1,\ldots,N\right\} $
be a set of states in $\mathbb{C}^{d_{1}}\otimes\mathbb{C}^{d_{2}}$,
with $d_{1}\leq d_{2}$ and where state $\vert\psi_{i}\rangle$ occurs
with probability $p_{i}$. Let $\sqrt{\lambda_{i}}$ be the maximum
Schmidt coefficient of $\vert\psi_{i}\rangle$, and assume that the
states are labeled so that $p_{1}\lambda_{1}\ge p_{2}\lambda_{2}\ge\cdots\ge p_{N}\lambda_{N}$.

Let $r$ be the positive integer such that ${\displaystyle \kappa:=\sum_{i=1}^{r-1}\lambda_{i}^{-1}\le d_{1}d_{2}<\sum_{i=1}^{r}\lambda_{i}^{-1}}$.
Then for any separable measurement $\mathbb{M}$,\begin{eqnarray}
P_{s}\left(\mathbb{M}\right) & \leq & \sum_{i=1}^{r-1}p_{i}+p_{r}\lambda_{r}\left(d_{1}d_{2}-\kappa\right)\label{thm-4-eq-1}\end{eqnarray}
 \end{thm} \begin{proof} For every measurement outcome $a$, we
assign a best guess $\mathbb{G}(a)\in\{\vert\psi_{i}\rangle\}$ of
the identity of our state. This partitions the set of measurement
outcomes, and we write $\mathbb{G}^{-1}(i)$ as the set of measurement
outcomes $a$ for which $\mathbb{G}(a)=\vert\psi_{i}\rangle$. Note
that for any $i$,\begin{eqnarray}
\sum_{a\in\mathbb{G}^{-1}(i)}m_{a}\left\vert \langle\psi_{i}\vert\chi_{a}\rangle\right\vert ^{2} & \leq & 1.\label{ineq measure}\end{eqnarray}
 In addition, Lemma 2 tells us that for any $a$, $\left\vert \langle\psi_{i}\vert\chi_{a}\rangle\right\vert ^{2}\le\lambda_{i}$.
If we write $\tau_{i}=\sum_{a\in\mathbb{G}^{-1}(i)}m_{a}$, then\begin{eqnarray}
\sum_{a\in\mathbb{G}^{-1}(i)}m_{a}\left\vert \langle\psi_{i}\vert\chi_{a}\rangle\right\vert ^{2} & \leq & \sum_{a\in\mathbb{G}^{-1}(i)}m_{a}\lambda_{i}=\lambda_{i}\tau_{i}\label{ineq measure 2}\end{eqnarray}
 Combining the two bounds, (\ref{ineq measure}) and (\ref{ineq measure 2}),
gives us

\begin{eqnarray*}
P_{s}\left(\mathbb{M}\right) & = & \sum_{a}m_{a}\max_{i}p_{i}\left|\langle\psi_{i}|\chi_{a}\rangle\right|^{2}\\
 & = & \sum_{i}p_{i}\sum_{a\in\mathbb{G}^{-1}(i)}m_{a}\left|\langle\psi_{i}|\chi_{a}\rangle\right|^{2}\\
 & \leq & \sum_{i}p_{i}\min\left(1,\lambda_{i}\tau_{i}\right)\\
 & = & \sum_{i}p_{i}\lambda_{i}\min\left(\lambda_{i}^{-1},\tau_{i}\right)\\
 & \le & \max_{\left\{ \tau_{i}\right\} }\left\{ \sum_{i}p_{i}\lambda_{i}\tau_{i}:\tau_{i}\in\left[0,\lambda_{i}^{-1}\right],\sum_{i}\tau_{i}=d_{1}d_{2}\right\} \end{eqnarray*}
 This constrained optimization problem is solved by making $\tau_{i}$
as large as possible for large values of $p_{i}\lambda_{i}$, until
you reach $\sum_{i}\tau_{i}=d_{1}d_{2}$. This gives the bound in
(\ref{thm-4-eq-1}). 
 \end{proof}

\section{Optimality of the Bounds: Examples }

\label{Sec-Examples}

In this section we will present examples where the bounds obtained
in the previous section are shown to be tight. In each case the optimality
of the bound in question follows by computing the separable fidelity
exactly. We also give explicit local strategies to achieve these values.

\subsection{Lower bound in Theorem 2 and upper bound in Theorem 4}

\subsubsection*{Example 1 }

Consider the set of four Bell states in $\mathbb{C}^{2}\otimes\mathbb{C}^{2}$,\begin{eqnarray*}
|\Phi_{1}\rangle & = & \frac{1}{\sqrt{2}}\left(|00\rangle+|11\rangle\right),\\
|\Phi_{2}\rangle & = & \frac{1}{\sqrt{2}}\left(|00\rangle-|11\rangle\right),\\
|\Phi_{3}\rangle & = & \frac{1}{\sqrt{2}}\left(|01\rangle+|10\rangle\right),\\
|\Phi_{4}\rangle & = & \frac{1}{\sqrt{2}}\left(|01\rangle-|10\rangle\right),\end{eqnarray*}
 with probabilities $p_{1}\geq p_{2}\geq p_{3}\geq p_{4}$.

The following facts are known: (a) no more than two Bell states can
be perfectly distinguished by LOCC \cite{Ghosh-2001} or by separable
measurements (Corollary 3, \cite{Cosentino-2012}), and (b) any two
Bell states can be perfectly distinguished by LOCC (this follows from
the result in \cite{Walgate-2000}). Thus the lower bound according
to Theorem 2 is given by $F_{S}\geq p_{1}+p_{2}$.

On the other hand, Theorem 4 implies that $P_{s}(S)\le p_{1}+p_{2}$,
as each of the $\lambda_{i}=\frac{1}{2}$. Since the Bell states are
orthogonal, $F_{S}=P_{s}(S)\le p_{1}+p_{2}$. Thus, it must be the
case that $F_{S}=p_{1}+p_{2}$. This matches the result obtained numerically
in \cite{Navascu\'es}.

To attain this fidelity by LOCC one can simply do the measurement
in the following product basis: $\left\{ |++\rangle,|+-\rangle,|-+\rangle,|--\rangle\right\} $,
where $|\pm\rangle=\frac{1}{\sqrt{2}}\left(|0\rangle\pm|1\rangle\right)$
followed by the decoding map: $\left\{ ++,--\right\} \rightarrow\Phi_{1};\left\{ +-,-+\right\} \rightarrow\Phi_{2}$.
Note that the measurement perfectly distinguishes the states $|\Phi_{1}\rangle,|\Phi_{2}\rangle$.

\subsubsection*{Example 2}

Consider the set of four GHZ states in $\mathbb{C}^{2}\otimes\mathbb{C}^{2}\otimes\mathbb{C}^{2}$:\begin{eqnarray*}
|\Psi_{1}\rangle & = & \frac{1}{\sqrt{2}}\left(|000\rangle+|111\rangle\right),\\
|\Psi_{2}\rangle & = & \frac{1}{\sqrt{2}}\left(|000\rangle-|111\rangle\right),\\
|\Psi_{3}\rangle & = & \frac{1}{\sqrt{2}}\left(|011\rangle+|100\rangle\right),\\
|\Psi_{4}\rangle & = & \frac{1}{\sqrt{2}}\left(|011\rangle-|100\rangle\right),\end{eqnarray*}
 with probabilities $p_{1}\geq p_{2}\geq p_{3}\geq p_{4}$.

Label the qubits as $A,$ $B$, and $C$. Observe that the set is
locally indistinguishable across the bipartition $A:BC$. This is
because in the bipartition $A:BC$ the states look exactly like the
four Bell states embedded in $\mathbb{C}^{2}\otimes\mathbb{C}^{4}$.
By the previous example, $F_{S}\left(A:BC\right)=p_{1}+p_{2}$. On
the other hand, the set is perfectly distinguishable across the bipartitions
$B:AC$ and $C:AB$. This implies that $F_{S}\left(B:CA\right)=F_{S}\left(C:AB\right)=1$.

However, the separable fidelity in a multipartite setting is bounded
by the minimum separable fidelity over all bipartitions. That is,\begin{eqnarray*}
F_{S}\left(A:B:C\right) & \leq & \min\left\{ F_{S}\left(A:BC\right),F_{S}\left(B:AC\right),F_{S}\left(C:AB\right)\right\} \\
 & \leq & F_{S}\left(A:BC\right)=p_{1}+p_{2}.\end{eqnarray*}
 For a lower bound, we know that any two orthogonal multipartite states
can be locally distinguished \cite{Walgate-2000}, so by Theorem 2,
$F_{S}(A:B:C)\ge P\left(\{|\Psi_{1}\ket,|\Psi_{2}\ket\}\right)=p_{1}+p_{2}$.
We then immediately obtain that $F_{S}\left(A:B:C\right)=p_{1}+p_{2}$.


\subsection{Upper bound in Corollary 2}

\subsubsection*{Example 3 }

We generalize Example 1 by considering the following orthogonal basis
in $\mathbb{C}^{2}\otimes\mathbb{C}^{2}$: \begin{eqnarray*}
|\psi_{1}\rangle & = & \alpha|00\rangle+\beta|11\rangle,\\
|\psi_{2}\rangle & = & \beta|00\rangle-\alpha|11\rangle,\\
|\psi_{3}\rangle & = & \alpha|01\rangle+\beta|10\rangle,\\
|\psi_{4}\rangle & = & \beta|01\rangle-\alpha|10\rangle,\end{eqnarray*}
 where $\alpha\geq\beta>0$ are real and satisfy $\alpha^{2}+\beta^{2}=1$.
The basis is known not to be perfectly distinguished by LOCC \cite{Ghosh-2002}
and cannot be distinguished by separable measurements either (Corollary
2). We consider the situation when the above states are equally likely.
It follows from Corollary 2 that $F_{S}\leq\alpha^{2}$. By measuring
in the computational basis, and using the decoding map: $00\rightarrow|\psi_{1}\rangle,11\rightarrow|\psi_{2}\rangle,01\rightarrow|\psi_{3}\rangle,10\rightarrow|\psi_{4}\rangle$,
one can easily compute the achievable fidelity, which comes out to
be $\alpha^{2}$. Therefore for the above set of equally likely states
$F_{S}=\alpha^{2}$.

\subsubsection*{Example 4}

Consider the following orthogonal basis of three qubits:\begin{eqnarray*}
|\psi_{1}\rangle & = & \alpha|000\rangle+\beta|111\rangle,\\
|\psi_{2}\rangle & = & \beta|000\rangle-\alpha|111\rangle,\\
|\psi_{3}\rangle & = & \alpha|001\rangle+\beta|110\rangle,\\
|\psi_{4}\rangle & = & \beta|110\rangle-\alpha|001\rangle,\\
|\psi_{5}\rangle & = & \alpha|011\rangle+\beta|100\rangle,\\
|\psi_{6}\rangle & = & \beta|011\rangle-\alpha|100\rangle,\\
|\psi_{7}\rangle & = & \alpha|010\rangle+\beta|101\rangle,\\
|\psi_{8}\rangle & = & \beta|010\rangle-\alpha|101\rangle,\end{eqnarray*}
 where $\alpha\geq\beta$ are real and satisfy $\alpha^{2}+\beta^{2}=1$.
We assume that all states are equally likely.

Let the qubits be labeled as $A,$ $B$ and $C$. The upper bound
in Corollary 2 cannot be directly applied because it holds for bipartite
systems. By inspection we see that across every bipartition (for example,
$A:BC$) each state has a maximum Schmidt coefficient of $\alpha$.
Therefore, we can apply Corollary 2 to get \begin{eqnarray*}
F_{S}\left(i:jk\right) & \leq & \alpha^{2}:i\neq j\neq k\in\left\{ A,B,C\right\} \end{eqnarray*}
 Noting that the separable fidelity in a multipartite setting is bounded
by the minimum separable fidelity across all bi-partitions, we have\begin{eqnarray*}
F_{S}\left(A:B:C\right) & \leq & \min\left\{ F_{S}\left(A:BC\right),F_{S}\left(B:AC\right),F_{S}\left(C:AB\right)\right\} \\
 & \leq & \alpha^{2}\end{eqnarray*}
 This upper bound is attainable by LOCC simply by measuring in the
computational basis and decoding with the most likely input, as in
Example 2. This succeeds with probability $\alpha^{2}$; since our
states are orthogonal, this implies that the fidelity $F_{S}=\alpha^{2}$.
The previous argument shows that this is optimal.

\subsection{Lower bound in Corollary 1 and upper bound in Corollary 3}

\subsubsection*{Example 5}

We consider distinguishing a set of states chosen from the canonical
maximally entangled basis in $\mathbb{C}^{d}\otimes\mathbb{C}^{d}$
\begin{equation}
|\Psi_{nm}\rangle=\frac{1}{\sqrt{d}}\sum_{j=0}^{d-1}e^{\frac{2\pi ijn}{d}}|j\rangle\otimes|\left(j+m\right)\mod d\rangle\label{Bell-d-general}\end{equation}
 for $n,m=0,1,\cdots,d-1$. The following facts are known: (a) any
set of $N$ orthogonal states chosen from the above set is not perfectly
distinguishable by separable measurements (therefore by LOCC) when
$N>d$; (b) the above basis can be grouped into $d$ subsets $S_{k}:k=0,...,d-1$,
where $S_{k}$ consists of the states $|\Psi_{0k}\rangle,|\Psi_{1k}\rangle,...,|\Psi_{(d-1)k}\rangle$
and such a subset can be perfectly distinguished by LOCC \cite{Ghosh-2004}.

Now construct any set of $N>d$ orthogonal states such that it contains
all states from at least one subset $S_{k}$ for some $k$. Assume
that all states are equally likely. From Corollary 1 it follows that
$F_{S}\geq d/N$. On the other hand, for any set of $N$ equally likely
orthogonal maximally entangled states in $\mathbb{C}^{d}\otimes\mathbb{C}^{d}$,
we have shown that (Corollary 3) $F_{S}\leq d/N$. Putting it all
together we have $F_{S}=d/N$. That this bound is also achieved by
LOCC follows by noting that the set contains all vectors from a perfectly
LOCC-distinguishable subset $S_{k}$ (for some $k$).

\subsection{Upper bound in Theorem 4}

\subsubsection*{Example 6}

Consider the following set of states in $\mathbb{C}^{3}\ot\mathbb{C}^{3}$.
Note that the first three are orthogonal maximally entangled states
in $\mathbb{C}^{3}\ot\mathbb{C}^{3}$, which implies that they are
locally distinguishable \cite{Nathanson-2005}. The last three states
are also orthogonal and perfectly distinguishable. \begin{eqnarray*}
|\psi_{0}\rangle & = & \frac{1}{\sqrt{3}}\left(|00\rangle+|11\rangle+|22\rangle\right),\\
|\psi_{1}\rangle & = & \frac{1}{\sqrt{3}}\left(|01\rangle+|12\rangle+|20\rangle\right),\\
|\psi_{2}\rangle & = & \frac{1}{\sqrt{3}}\left(|02\rangle+|10\rangle+|21\rangle\right),\\
|\psi_{3}\rangle & = & |00\rangle.\end{eqnarray*}
 Suppose that the three maximally entangled states each occur with
probability $p$, and the product state $\vert\psi_{3}\rangle$ occurs
with probability $q=1-3p$. Since the entangled states are perfectly
distinguishable, Theorem 1 says that the probability of successful
identification is at least $3p$. On the other hand, if we measure
in the computational basis, then we will always successfully identify
our state unless we get the result $00$. In this case, we see that
$P(\vert\psi_{0}\rangle|00)=\frac{p}{3q+p}$ and $P(\vert\psi_{3}\rangle|00)=\frac{3q}{3q+p}$.
We choose the maximum likelihood answer, which means that the optimal
separable probability of success $P_{s}(S)$ is at least $1-\min(p/3,q)$.
We claim that this error probability is optimal.

Theorem 4 requires us to sort the quantities $\{p_{i}\lambda_{i}\}$
into decreasing order. If $q\le p/3$, then we begin with the three
entangled states and get $\lambda_{0}^{-1}+\lambda_{1}^{-1}+\lambda_{2}^{-1}=9$,
which is the dimension of the space. This implies that $P_{s}\le3p=1-q$.
On the other hand, if $q>p/3$, it is the product state which maximizes
$p_{i}\lambda_{i}$ and we see that $\lambda_{3}^{-1}+\lambda_{1}^{-1}+\lambda_{2}^{-1}<9<\lambda_{3}^{-1}+\lambda_{1}^{-1}+\lambda_{2}^{-1}+\lambda_{0}^{-1}$.
In this case, Theorem 4 yields $P_{s}\le q+2p+\frac{p}{3}\left(2\right)=1-\frac{p}{3}$
since we need to use a fraction of the fourth term in our sum. In
both cases, we see that the upper bound from the theorem matches the
achievable lower bound with the computational basis.

It is a little surprising that the bound from Theorem 4 is tight in
this case, since it only makes use of the maximal Schmidt coefficient
{and} does not use the fact that the states are nonorthogonal. The
example from \cite{Horodecki-2003} shows that there is no direct
correlation between the entanglement and the probability of discrimination;
thus Theorem 4 will not give tight bounds in general. The issue of
states which are not orthogonal is raised in the next section.

\section{Nonorthogonal states: Probability of Success vs Separable Fidelity}

\label{Sec-QubitExamples}

When we attempt to distinguish quantum states which are not all orthogonal,
we face two challenges: the overlap between the states and the restriction
to separable measurements. As an illustration of this phenomenon,
we examine the following sets in $\mathbb{C}^{2}\otimes\mathbb{C}^{2}$:\begin{align*}
S_{1} & =\left\{ |\phi_{i}\rangle,i=0,1,2,3\right\} =\left\{ |01\rangle,|10\rangle,\frac{1}{\sqrt{2}}\left(|00\rangle+|11\rangle\right),\frac{1}{\sqrt{2}}\left(|00\rangle-|11\rangle\right)\right\} ,\\
S_{2} & =\left\{ |\psi_{i}\rangle,i=0,1,2,3\right\} =\left\{ |00\rangle,|11\rangle,\frac{1}{\sqrt{2}}\left(|00\rangle+|11\rangle\right),\frac{1}{\sqrt{2}}\left(|00\rangle-|11\rangle\right)\right\} .\end{align*}
 Note that $S_{1}$ is a complete orthonormal basis which can be perfectly
distinguished in the full space, while the span of $S_{2}$ is only
two-dimensional, which severely limits the mutual information between
the identity of our state and the outcome of our measurement. In each
case, we assume the four states are equally probable.

In $S_{1}$, the Schmidt coefficients are $1,1,\frac{1}{\sqrt{2}},\frac{1}{\sqrt{2}}$.
Since $1+1+\frac{1}{1/2}=4$, Theorem 4 tells us that $P_{s}(S)\le\frac{3}{4}$,
and this is achieved by measuring in the computational basis. According
to Theorem 1, the separable fidelity is also equal to $\frac{3}{4}$
since the states are orthogonal.

On the other hand, $S_{2}$ is highly dependent. We write $Q=|00\rangle\langle00|+|11\rangle\langle11|$
as the projection onto the two dimensional span of $S_{2}$ and note
that $Q=\vert\psi_{0}\rangle\langle\psi_{0}\vert+\vert\psi_{1}\rangle\langle\psi_{1}\vert=\vert\psi_{2}\rangle\langle\psi_{2}\vert+\vert\psi_{3}\rangle\langle\psi_{3}\vert$.
As a result, for any matrix $M$, $\langle\psi_{0}\vert M\vert\psi_{0}\rangle+\langle\psi_{1}\vert M\vert\psi_{1}\rangle=\langle\psi_{2}\vert M\vert\psi_{2}\rangle+\langle\psi_{3}\vert M\vert\psi_{3}\rangle=\mbox{Tr}QM$.
This implies that it is impossible to gain any information about whether
our state comes from $\{\vert\psi_{0}\ket,\vert\psi_{1}\ket\}$ or
from $\{\vert\psi_{2}\ket,\vert\psi_{3}\ket\}$, even if we are allowed
to measure across the full space:\[
P_{s}\left(\mathbb{M}\right)=\frac{1}{4}\sum_{j}\max_{i}\langle\psi_{i}|M_{j}|\psi_{i}\rangle\leq\frac{1}{4}\sum_{j}\mbox{Tr}QM_{j}=\frac{1}{4}\tr Q=\frac{1}{2}.\]
 This upper bound is attained by simply assuming that the state is
either $\vert\psi_{0}\rangle$ or $\vert\psi_{1}\rangle$ and optimally
distinguishing them. This can also be accomplished with one-way LOCC,
so in this case our separable probability is equal to the global probability:\[
P_{s}(S)=P_{s}=\frac{1}{2}.\]
 To calculate the fidelity, we introduce the following lemma, which
applies to any measurement $\mathbb{M}$ (separable or not) and is
useful when there is a linear dependence among the possible states.

\begin{lem}Given an ensemble of states, $\{p_{i},\vert\psi_{i}\rangle\}$
such that the linear span of the states $\vert\psi_{i}\rangle$ has
dimension $r$, the average fidelity of a protocol $(\mathbb{M},\mathbb{G})$
is bounded by\begin{eqnarray*}
F\left(\mathbb{M},\mathbb{G}\right) & \leq & \left\Vert \rho^{\prime}\right\Vert _{r}^{KF},\end{eqnarray*}
 where $\rho'=\sum_{i}p_{i}\vert\psi_{i}\rangle\langle\psi_{i}\vert\otimes\vert\psi_{i}\rangle\langle\psi_{i}\vert$
and $||\cdot||_{r}^{KF}$ is the Ky Fan norm and is simply the sum
of the first $r$ singular values.\end{lem}

Note that if $\vert\psi_{i}\rangle$ are linearly independent, then
$||\rho'||_{r}^{KF}=\mbox{Tr}\rho'=1$, but if they are dependent,
this can give a nice bound. \begin{proof} Let $\mathbb{M}=\{m_{a}\vert\chi_{a}\rangle\langle\chi_{a}\vert\}$
and $\mathbb{G}(a)=\vert K_{a}\rangle$ and let $Q$ be the projection
onto the span of $\{\vert\psi_{i}\rangle\}$.

Define the matrix $M=\sum_{a}m_{a}Q\vert\chi_{a}\rangle\langle\chi_{a}\vert Q\otimes\vert K_{a}\rangle\langle K_{a}\vert$
and write\begin{eqnarray*}
F\left(\mathbb{M},\mathbb{G}\right) & = & \sum_{a}\sum_{i}p_{i}m_{a}\left|\langle\psi_{i}|\chi_{a}\rangle\right|^{2}\left|\langle\psi_{i}|K_{a}\rangle\right|^{2}=\mbox{Tr}\rho^{\prime}M\end{eqnarray*}
 Noting that $\mbox{Tr}M=\mbox{Tr}Q=r$ and $||M||_{\infty}\le1$,
we see that the maximal value of $\mbox{Tr}\rho'M$ is the sum of
the $r$ maximum eigenvalues of the positive semidefinite matrix $\rho'$.
This proves the lemma. \end{proof} We can apply this lemma to our
set $S_{2}$, whose span has dimension 2. While the matrix $\rho=\frac{1}{4}\sum_{i}\vert\psi_{i}\rangle\langle\psi_{i}\vert$
has an eigenvalue $\frac{1}{2}$ with multiplicity 2, $\rho'=\frac{1}{4}\sum_{i}\vert\psi_{i}\rangle\langle\psi_{i}\vert\otimes\vert\psi_{i}\rangle\langle\psi_{i}\vert$
has eigenvalues $\{\frac{1}{2},\frac{1}{4},\frac{1}{4},0\}$, which
means that $||\rho'||_{2}^{KF}=\frac{3}{4}$ and the average fidelity
$F$ is at most three quarters.

This bound can be achieved by projecting $S_{2}$ onto the computational
basis, which can be implemented locally. Thus, for the set $S_{2}$,\[
F=F_{S}=\frac{3}{4}\]
 Thus, there is no difference between global and local measurements
for $P_{s}$ and $F$ with respect to the linearly dependent set $S_{2}$.

Note that the separable fidelity is the same for $S_{1}$ and $S_{2}$
even though the success probabilities are different. Although $S_{1}$
consists of four mutually orthogonal states while the four states
of $S_{2}$ are coplanar, the separable fidelity sees the problems
as equally challenging. This highlights the fact that having states
close together makes approximate cloning easier, increasing the fidelity,
but makes state identification harder, decreasing the success probability.
For both measures, in the shift from $S_{1}$ to $S_{2}$, as the
overlap between the states $\{\vert\psi_{i}\ket\}$ grows, the gap
between separable and global measurements shrinks, which is consistent
with previously known results (such as \cite{Nathanson-2010}).

\section{Conclusions}

\label{Sec-Conclusions}

Local distinguishability of orthogonal states has been the subject
of intensive research in the last decade as it allows us to explore
foundational concepts of quantum theory and quantum information. These
include entanglement and nonlocality as well as the potentials and
limitations of LOCC protocols. In this paper we have addressed a basic
question, which is how much can be discovered about a given quantum
system using a separable measurement. We have obtained lower and upper
bounds on the separable fidelity, the optimal average fidelity based
on information obtained by a separable measurement, and have given
examples in both bipartite and multipartite settings where these bounds
are optimal. We have also shown that, if our initial states are orthogonal,
a strategy that minimizes the error probability is necessarily optimal
for separable fidelity. These general bounds are useful, as explicit
expressions for fidelity and success probability are hard to find
even in specific cases.

There remain many open problems in the area of local discrimination
and the relationship between separable and local operations. We have
established that if our set of possible states is orthogonal, then
$P_{s}(S)$ and $F_{S}$ are equal but that, in general, they diverge
with nonorthogonal states. It would be useful to quantify this complementarity
relation in the separable realm. A direct analog of Theorem 4 applying
to separable fidelity would help in this direction. There also remains
much work in understanding the gap between optimal global measurements
and optimal separable measurements in the presence of non-orthogonality,
which seems to affect global bounds faster than separable ones. Finally,
we look forward to understanding the implications of these bounds
in the asymptotic context of many copies of our multipartite systems.

\medskip{}

\noindent \textbf{Acknowledgment:} The second author is grateful to
Saint Mary's College for granting a sabbatical, during which this
work was completed.

\end{document}